\documentclass[twocolumn,pra,aps,amsmath,amsfonts,amssymb,superscriptaddress,showpacs]{revtex4-1}

\usepackage{natbib}
\usepackage{hyperref}
\usepackage{amsthm}
\usepackage{graphicx}
\usepackage{bm}
\usepackage{dsfont}

\theoremstyle{definition}

\newtheorem{thm}{Theorem}

\newtheorem{prop}[thm]{Proposition}
\newtheorem{corr}[thm]{Corollary}

\newcommand{\bra}[1]{\langle #1 \rvert}
\newcommand{\ket}[1]{\lvert #1 \rangle}

\newcommand{\projop}[1]{\ket{#1} \bra{#1}}

\newcommand{\defeq}{:=}

\newcommand{\F}{\mathbb{F}}

\newcommand{\with}{\mid}

\newcommand{\one}{\mathds{1}}

\DeclareMathOperator{\Tr}{Tr}

\begin{document}

\title{Quantum secret sharing with qudit graph states}
\author{Adrian Keet}
\author{Ben Fortescue}
\affiliation{%
Institute for Quantum Information Science,\\
University of Calgary, Alberta T2N 1N4, CANADA\\
}%
\author{Damian Markham}
\affiliation{
LTCI-CNRS, Telecom ParisTech,
37/39 rue Dareau, 75014 Paris, France}
\author{Barry C. Sanders}
\affiliation{%
Institute for Quantum Information Science,\\
University of Calgary, Alberta T2N 1N4, CANADA\\
}%

\date{December 21, 2010}

\begin{abstract}
We present a unified formalism for threshold quantum secret sharing using graph states of systems with prime dimension.  We construct protocols
for three varieties of secret sharing: with classical and quantum secrets shared between parties
over both classical and quantum channels.
\end{abstract}

\maketitle

\section{Introduction}
``Secret sharing'' refers to an important family of multiparty cryptographic protocols in
 both the classical \cite{Sham79, Schneier96} and the quantum \cite{HBB99,CGL99,Gott00,MS08} contexts.
A secret sharing protocol
comprises a dealer and $n$ players who are interconnected by some set of classical or quantum channels.
The ``secret'' to be shared is a classical string or quantum state and is distributed among the players by the dealer
in such a way that it can only be recovered by certain subsets of players acting collaboratively.  The {\it access structure}
is the set of all subsets of players who can recover the secret, and the {\it adversary structure} corresponds to those subsets that obtain no knowledge of the secret.
There may, in addition,
be external eavesdroppers who should also gain no knowledge of the secret.
Secret sharing protocols have practical
applications in, for example, money transfer \cite{CGL99} and voting \cite{schoen99} schemes.
Recently, a unified formalism for a range of protocols for sharing both
classical and quantum secrets using qubit graph states \cite{HDER05} has been proposed \cite{MS08}.  The formalism allows for both classical and quantum
channels between the dealer and players and covers three common varieties of secret sharing within a single framework.  In this paper we develop
an analogous formalism for the case of systems of prime-dimensional quantum states.

Although qubit systems are generally the most straightforward to describe and the theory of qubit graph states is well-developed, there are good
reasons to consider higher-dimensional systems as well.  If players' shares consist of such systems, it may be more efficient to make direct
use of the larger Hilbert space than to work with encoded qubits.  Additionally, certain access structures are known not to be possible 
using the existing secret sharing schemes on qubit graph states \cite{MS08,Kashefi09} or, in some cases, using any scheme with qubit shares.  An example
of the latter (which can, however, be implemented using qutrits) was given by Cleve, Gottesman, and Lo \cite{CGL99}.
A more broad justification is that one often needs to consider higher-dimensional states in order to derive rigorous general results in quantum information theory,
such as in certain {\it no-go theorems} \cite{Shaari08, Chen08, NFC09}.

Our aim is therefore to find a unified formalism for secret sharing that is not restricted to the qubit case.
As discussed recently \cite{HDER05, KKKS05, BB06}, graph and stabilizer states can be extended to the qudit case for prime dimensions.
Following this approach, we find direct extensions of the existing qubit protocols \cite{MS08} in $d$ dimensions, where $d$ is prime i.e.\ we find protocols for sharing classical secrets with $d$ possible values and secret quantum states within a $d$-dimensional Hilbert space, using graph states shared
by a dealer and players, each with a $d$-dimensional subsystem.  Our work therefore achieves, in higher-dimensional systems, the goals previously achieved for qubits \cite{MS08}, namely, providing a general graph-state unification for sharing classical and quantum secrets using both classical and quantum channels, and in this way extends the protocols to more general access structures.  Note that we are primarily demonstrating how a subset of existing secret sharing protocols, already known to be achievable, can be unified
within our formalism; we are not attempting to address via the formalism any existing limitations in quantum and classical secret sharing.
Additionally, however, we have demonstrated protocols within our formalism which have not been previously shown: the case of three parties sharing
classical secrets distributed over quantum channels (secure or insecure), in which a minimum of two parties are required to recover the secret.

In general, there are many different possible access structures for secret sharing schemes.
In our work we consider the specific case of {\it threshold} secret sharing.  In such a scheme the secret can be recovered if and only if any
$k$ of the $n$ players collaborate to do so (and any set of fewer than $k$ players are denied any knowledge of the secret).
All that is required is that enough players collaborate;
i.e.\ it does not matter which $k$ players do so.  Thus the access structure comprises any set of at least $k$ players, with the remainder forming the adversary structure.
This is denoted as a $(k,n)$ threshold secret sharing scheme.  Note that one can construct arbitrary access structures for both classical \cite{BL88} and quantum \cite{Gott00} secret sharing by concatenating threshold protocols; thus it is sufficient to consider only threshold secret sharing to address general secret sharing.

Here we consider three specific varieties of such schemes previously demonstrated in qubit graph states \cite{MS08}.
We note that all existing forms of secret sharing that have been proposed fall into one of these categories.
\begin{enumerate}
\item CC schemes: The secret is classical, the dealer is connected to the players via private quantum channels and all players are connected by private classical channels.
\item CQ schemes: The secret is classical, the dealer shares public quantum channels with each player and the players
are connected to each other by private classical channels.
\item QQ schemes: The secret is quantum, the dealer shares either private or public quantum channels with each player
and the players are connected to each other by private quantum or classical channels.
\end{enumerate}

In Sec.\ \ref{sec-gstates}, we formally define the structure and graphical representation of graph states
for qudits and explain how classical information can be encoded on to such states and later accessed.
In Sec.\ \ref{sec-locmeas} we show the effect of local ($d$-dimensional) Pauli measurements
on qudit graphs and derive rules (partially shown previously \cite{BB06}) for the
form of the consequent-reduced graph states, analogous to ``local complementation''
in the qubit case.
In Sec.\ \ref{sec-protocols} we demonstrate specific CC, CQ and QQ protocols which may be implemented
in qudit graphs given the properties previously derived (though we note that in the case of the $(n,n)$ QQ scheme, the secret is not perfectly
denied to the adversary structure).  We present, as an example of the power of higher dimensions, the $(2,3)$ case, which is
 not possible with qubit graph state schemes \cite{MS08}, as well as higher-dimensional analogues to those presented in \cite{MS08}.

\section{Graph states for qudits}\label{sec-gstates}
Graph states \cite{BR01} are a class of entangled multipartite states (including the well-known $n$-qubit GHZ state) of wide interest in quantum information.
  Their useful properties include
a convenient graphical representation and characterisation within the stabilizer formalism \cite{Gott97}, as well as their practical application
within information processing, in particular for computation \cite{RB01} and error correction \cite{BB06,Looi08,Beigi09}, and their amenability for the representation of information flow \cite{Kashefi09,Gheorghiu07}. Such states have been created and applied in this way for up to six qubits \cite{WRR05, LGZ07}, and recent work \cite{BB06}
has considered the case of higher-dimensional graph states.  We employ the same framework in sections \ref{sec-ops} and \ref{sec-gs} as that in previous work
\cite{HDER05,KKKS05, BB06} although we use different notation more suited to secret sharing.

\subsection{Graph states}\label{sec-gs}
We now introduce qudit graph states and their graphical representation.
Consider a  finite field $\F_d$ of order $d$, where $d$ is a prime number $>2$, and a
weighted undirected graph i.e.\ a set of $n$ vertices $\mathsf{V} = \{\mathsf{v}_i\}$ joined by
a set of edges $\mathsf{E}=\{\mathsf{e}_{ij}=\{\mathsf{v}_i,\mathsf{v}_j\}\}$.

Each edge $\mathsf{e}_{ij}$ has an assigned ``weight'' $A_{ij}\in \F_d$.  $\mathsf{E}$ only contains
edges with non-zero weights: a weight of zero corresponds to two vertices not being joined
by an edge.  No vertex is joined to itself.  We can summarise this information in an
$n\times n$ ``adjacency matrix'' $A$ with elements $A_{ij}$.
We therefore have $A_{ij}=A_{ji}$ for all $\mathsf{v}_i,\mathsf{v}_j\in \mathsf{V}$
and $A_{ii}=0$ for all $\mathsf{v}_i \in \mathsf{V}$.

We define the computational basis $\{ \ket{j} \with j \in \F_d \}$.
Then the graph state
represented by the above graphical construction is
\begin{equation}
\ket{G}\defeq\prod_{\mathsf{e}_{ij}\in \mathsf{E}}C_{ij}^{A_{ij}}\ket{\overline{0}}^{\otimes n}
\end{equation}
where we define the two-qudit controlled-$Z$ operator
\begin{equation}
C_{ab} \ket{j}_a \ket{k}_b \defeq \omega^{jk} \ket{j}_a \ket{k}_b,
\end{equation}
the operator
\begin{equation}
U \vert i\rangle := \sum_{j \in \mathbb{F}_d} \omega^{ij}
\vert{j}\rangle,
\end{equation}
and the basis
\begin{equation}
\vert\bar{i}\rangle := U^{-1} \vert{i}\rangle \mid i \in \mathbb{F}_d.
\end{equation}

An example graph and associated adjacency matrix are depicted in Fig.\ \ref{fig-graphexample}.
\begin{figure}
\includegraphics[height=4cm]{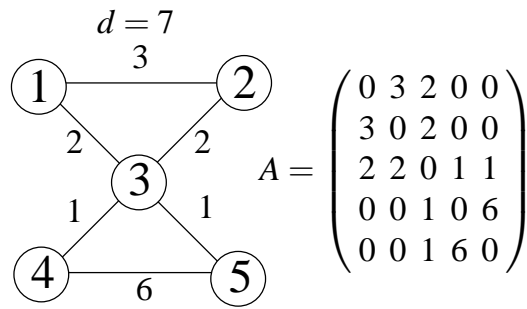}
\caption{A 5-qudit graph in $d=7$ and associated adjacency matrix $A$}\label{fig-graphexample}
\end{figure}

\subsection{Local operators and labelling}\label{sec-ops}
Here, in an analogous way to the qubit case \cite{MS08}, we describe the encoding of classical information onto the graph state by applying local operators
to the individual qudits that are graphically represented by the vertices of the graph.  Here and later we make use of the 
 {\it generalized Pauli  operators} \cite{Pat88,GKP01, BGS02}
\begin{align}
    Z \ket{j} &\defeq \omega^j \ket{j}, \\
    X \ket{j} &\defeq \ket{j + 1}
\end{align}
where $\omega = e^{2\pi i/d}$.  These operators satisfy $X^d=Z^d=I$.
(It follows that $ZX=\omega XZ$ and that the eigenstates of
the operators $Z, X, XZ, XZ^2\ldots XZ^{d-1}$ form a set of mutually unbiased bases \cite{BBRV02}.
We consider only this subset of the general operators $X^mZ^n$ as they are
sufficient for defining the graph-state stabilizers and for our secret sharing protocols.)
We further define the operator 
\begin{equation}
S\ket{j}\defeq\omega^{j(j-1)/2}\ket{j}.
\end{equation}
which satisfies $SZ=ZS$ and $SXS^{-1} = XZ$.

In a labelled graph state each vertex $\mathsf{v}_i$ is assigned a label $\ell_{i}\defeq (z_i, x_i, m_i)$, where $z_i, x_i, m_i \in \F_d$.
Denoting an operator $O$ acting on a vertex $\mathsf{v}_i$ as $O_i$, the graph state is labelled by applying the operators $S_i^{m_i}X_i^{x_i}Z_i^{z_i}$ to the state
$\ket{G}$.  We describe the combined labels for all vertices
by $\bm{\ell} \defeq (\bm{z}, \bm{x}, \bm{m})$, where $\bm{z} \defeq (z_1, z_2\dots , z_n)$,
and similarly for $\bm{x}$ and $\bm{m}$.
We use the notation
\begin{equation}
Z^{\bm{z}}\defeq Z_1^{z_1}\otimes Z_2^{z_2}\otimes\ldots \otimes Z_n^{z_n}
\end{equation}
and similarly for $X$ and $S$.  Thus the labelled graph state is
\begin{equation}
\ket{G_{\bm{\ell}}} \defeq S^{\bm{m}} X^{\bm{x}} Z^{\bm{z}} \ket{G}.\label{eq-labeldef}
\end{equation}
Note that our definition differs from that used previously for qubits \cite{MS08} in that we apply the $S$ operators after the $X$ operators
rather than before.

Finally, we define an {\it encoded graph state} as a labelled graph state with $\bm{x} = \bm{0}$.

\subsection{Stabilizers}
Higher-dimensional graph states can be represented within the stabilizer formalism \cite{HDER05, KKKS05, BB06}.
The encoded graph state with labels $\bm{\ell}$ satisfies, for each vertex $\mathsf{v}_i$ \cite{BB06},
\begin{equation}
K_i \ket{G_{\bm{\ell}}} = \omega^{-z_i} \ket{G_{\bm{\ell}}}\label{eq-stab}
\end{equation}
where
\begin{equation}
K_i\defeq (XZ^{m_i})_i Z^{\bm{A}_i},\label{eq-kadef}
\end{equation}
and we define, for a qudit operator $O$ and scalar multiple $k$,
\begin{equation}
O^{k\bm{A}_i}\defeq \prod_{\mathsf{v}_{j}\in \mathsf{V}}O_j^{kA_{ij}}.
\end{equation}

The $\{K_i\}$ is thus a set of stabilizers for the encoded graph state,
and $\ket{G_{\bm{\ell}}}$ is the unique state (up to a global phase) satisfying
(\ref{eq-stab}) for all vertices $v_i$.  Note that the $\{K_i\}$ are tensor products
of local operators; thus, even in the absence of any quantum communication,
the value of $K_i$ may be measured through local operations and classical communication (LOCC)
by the appropriate
subset of parties combining their local measurement results.  Note also
that we need only consider exponents of these local operators modulo $d$.

It follows from (\ref{eq-kadef}) and (\ref{eq-labeldef}) and the commutation relations of Sec.\ \ref{sec-ops}
that applying the stabilizer $K_i$ to a labelled graph state $\ket{G_{\bm{\ell}}}$ is
equivalent (up to a global phase) to implementing the label change $x_i \to x_i+1$
, $z_j \to z_j + A_{ij}$ (mod $d$) for all $\mathsf{e}_{ij}\in \mathsf{E}$; i.e.\
\begin{equation}
\ket{G_{\bm{\ell}}}=S^{\bm{m}} X^{\bm{x}} Z^{\bm{z}} \ket{G}\propto S^{\bm{m}} X^{\bm{x}}X_i Z^{\bm{z}}Z^{\bm{A}_i} \ket{G}
\end{equation}
Thus the two different labellings correspond to the
same physical state.
As in the qubit case \cite{MS08}, we can exploit this property
to determine the state's dependence on particular qudit labels.

\subsection{Local measurements}\label{sec-locmeas}
Our CQ and QQ protocols involve one of the parties as a designated ``dealer'',
whose local measurement on her own qudit produces a reduced graph state shared by the other parties
with some access structure and labelling determined by the dealer's measurement outcome.  We therefore
consider the effect of local measurements on qudit graph states.  This has previously been determined in the qubit case.

We will closely follow the derivation for the effect of measurements on qubit graphs
in \cite{HDER05}, the results of which were used in \cite{MS08}.  The results
below for the effect of local measurement on the adjacency matrix were previously derived \cite{BB06}
in the context of local complementation \cite{VdNDM04, HDER05}, although here we also explicitly derive the effect
of local measurement on the graph labels.
We find the following:

\begin{prop}
    \label{prop:measure2}
    Given an encoded graph state $\ket{G_{\bm{\ell}}}$, suppose a measurement of
    $(X^m Z)_i$
    yields the value $\omega^s$. Then the resultant labelled graph state
    is obtained by the following procedure:
    \begin{enumerate}
        \item[(M1)]
            For each pair $j, k$ of distinct neighbours of $i$, change $A_{jk}
            \mapsto A_{jk} + m A_{ij} A_{ik}$.
        \item[(M2)]
            Relabel each vertex $j$ by $z_j \mapsto z_j + A_{ij} s + m A_{ij}
            z_i + m A_{ij} (A_{ij} + 1) / 2$ and $m_j \mapsto m A_{ij}^2$.
        \item[(M3)]
            Remove vertex $i$ and corresponding edges.
    \end{enumerate}
\end{prop}

\begin{proof}
We first introduce the operators $U$ and $R$, defined as
\begin{align}
U \ket{k} &\defeq \sum_{j \in \F_d} \omega^{jk} \ket{j}\\
R&\defeq  U^{-1}S^{-1}U.
\end{align}
It follows that $UXU^{-1}= Z$, and $UZU^{-1} = X^{-1}$.
We also have $[R,X]=[S,Z]=0$, and $RZR^{-1}=SXS^{-1}=XZ$; i.e.
the $R$ operator performs a $Z\leftrightarrow XZ$ basis transformation
just as $S$ performs $X\leftrightarrow XZ$.

For each local Pauli operator $O = X^n Z^m$, we define the operator
$P_{O,j}$, which is the projection onto the $\omega^j$-eigenspace of $O$:
\begin{equation}
P_{O,j} \defeq \frac{1}{d} \sum_{k \in F_d} \omega^{-jk} O^k.
\end{equation}

We first consider the case $m = 0$ where the projective measurement is of $Z_i$. We have
\begin{equation}
\ket{G_{\bm{\ell}}} = \prod_{\mathsf{v}_j\in \mathsf{V}} C_{ij}^{A_{ij}} \ket{\bar 0}_i
    \ket{G_{\bm{\ell}} \setminus \mathsf{v}_i}_{{\mathsf{V} \setminus\mathsf{v}_i}}.
\end{equation}
where $\ket{G_\ell\setminus\mathsf{v}_i}$ denotes the graph state with vertex $\mathsf{v}_i$ and all its edges removed, and  $\mathsf{V} \setminus\mathsf{v}_i$ denotes the
set all of players excluding player $i$.
    Next, noting that
\begin{equation}
\prod_{\mathsf{v}_{j}\in \mathsf{V}}C_{ij}^{A_{ij}} = \prod_{\mathsf{v}_{j}\in \mathsf{V}} \sum_{k \in \F_d}
    (P_{Z,k})_i Z_j^{kA_{ij}} = \sum_{k \in \F_d} (P_{Z,k})_i Z^{k \bm{A}_i},
\end{equation}
    we have
    \begin{align}
        (P_{Z,s})_i \ket{G_{\bm{\ell}}}
        &=
        (P_{Z,s})_i\prod_{\mathsf{v}_{j}\in \mathsf{V}} C_{ij}^{A_{ij}} \ket{\bar 0}_i\ket{G_{\bm{\ell}}\setminus \mathsf{v}_i}_{\mathsf{V} \setminus \mathsf{v}_i}\nonumber\\
        &=
        (P_{Z,s})_i \Bigl( \sum_{k \in \F_d} (P_{Z,k})_i Z^{k \bm{A}_i} \Bigr) \ket{\bar 0}_i \ket{G_{\bm{\ell}} \setminus \mathsf{v}_i}_{\mathsf{V} \setminus \mathsf{v}_i}\nonumber\\
        &=
        \frac{1}{\sqrt{d}} \ket{s}_i Z^{s\bm{A}_i} \ket{G_{\bm{\ell}} \setminus \mathsf{v}_i}_{\mathsf{V}\setminus \mathsf{v}_i}.
    \end{align}

Now we consider the case $m \ne 0$. Let
\begin{equation}
L_i \defeq R^{-1}_i S^{-\bm{A}_i^2}.
\end{equation}
where we define
\begin{equation}
O^{k\bm{A}_i^2}\defeq \prod_{\mathsf{v}_{j}\in \mathsf{V}}O_j^{kA_{ij}^2}.
\end{equation}

    Then, for each $\mathsf{v}_j \in \mathsf{V}$,
    \begin{align}
        L_i^m K_j L_i^{-m}
        &= R_i^{-m} S^{-m \bm{A}_i^2} X_j Z^{\bm{A_j}} S^{m \bm{A}_i^2} R_i^m\nonumber\\
        &= \omega^{mA_{ij}(A_{ij}+1)/2} K_i^{-mA_{ij}} K_j'\label{eq-mzero},
    \end{align}
    where $K_j' = X_jZ^{\bm{A_j} +mA_{ij}(\bm{A}_i)}Z_j^{-mA_{ij}^2}$.  Equation (\ref{eq-mzero}) follows from the relation
    \begin{equation}
(X^{-1}Z)^{A_{ij}} = \omega^{A_{ij}(A_{ij}-1)/2} X^{-A_{ij}} Z^{A_{ij}}.
\end{equation}
    Thus, noting that $K_i$ commutes with $L_i$, we have
\begin{equation}
K_j' = \omega^{-mA_{ij}(A_{ij}+1)/2} L_i^m K_i^{mA_{ij}} K_j L_i^{-m},
\end{equation}
    so
    \begin{align}
        K_j' L_i^m \ket{G_{\bm{\ell}}}
        &= \omega^{-mA_{ij}(A_{ij}+1)/2} L_i^m K_i^{mA_{ij}} K_j \ket{G_{\bm{\ell}}} \nonumber\\
        &= \omega^{-mA_{ij}(A_{ij}+1)/2 - mA_{ij}z_i - z_j} L_i^m \ket{G_{\bm{\ell}}}.
    \end{align}
    The state $L_i^m \ket{G_{\bm{\ell}}}$, which we will denote $\ket{\tau_i^m(G_{\bm{\ell}})}$,
    is therefore an encoded graph state
    $\ket{\tau_i^m(G_{\bm{\ell}})}$, stabilized by the $K_j'$ operators.  The structure (i.e.\ the adjacency matrix)
    of the graph $\tau_a^m(G_{\bm{\ell}})$ is given by equating the $\bm{A_j'}$ for each vertex
    to the $\bm{ A_i}$ given in the standard form for graph-state stabilizers in (\ref{eq-kadef}), and we similarly
    find the labels
\begin{equation}
z_j' = z_j +m A_{ij} z_i + m A_{ij}(A_{ij}+1)/2.
\end{equation}

    We also have
    \begin{align}
        L_i^{-m} (P_{Z,s})_i L_i^m
        &= \frac{1}{d} \sum_{k \in \F_d} \omega^{-sk} R_i^m S^{m
        \bm{A}_i^2} Z_i^k S^{-m \bm{A}_i^2} R_i^{-m} \nonumber\\
        &= \frac{1}{d} \sum_{k \in \F_d} \omega^{-sk} (X^m Z)_i^k \nonumber\\
        &= (P_{X^m Z,s})_i.
    \end{align}
    Thus we compute
    \begin{align}
        &(P_{X^m Z,s})_i \ket{G_{\bm{\ell}}}\nonumber\\
        &= L_i^{-m} (P_{Z,s})_i L_i^m \ket{G_{\bm{\ell}}} \nonumber\\
        &= \frac{1}{\sqrt{d}} (R^m \ket{s})_i
        (S^{m \bm{A}_i^2} Z^{s \bm{A}_i} \ket{\tau_i^m(G_{\bm{\ell}}) \setminus
        \mathsf{v}_i})_{\mathsf{V} \setminus \mathsf{v}_i}.
    \end{align}
Hence a local measurement of $XZ^m$ on the vertex $\mathsf{v}_i$ gives the graph $\ket{\tau_i^m(G_{\bm{\ell}})}$
with the $\mathsf{v}_i$ vertex removed and additional labelling operators $S^{m \bm{A}_i^2} Z^{s \bm{A}_i}$ applied. 
\end{proof}

The reduced-state graph transformation given in M1 can be regarded as an analogue of local complementation for qubit graphs; note
however that we have not shown (as is true for the qubit case) that all graphs equivalent under local Clifford group operations
are similarly equivalent under such transformations.  This has been proven elsewhere \cite{BB06}.

An example of the effect of local measurement is given in Fig.\ \ref{fig-square}
for a 4-party square state with $d=5$, all edges of weight 2 and each vertex having the
labels $(z,x,m)=(1,0,1)$.  If player $1$ performs a measurement
of $XZ^2$ and gets a result of $\omega^2$, the state is transformed as shown.

\begin{figure}
\includegraphics[height=4cm]{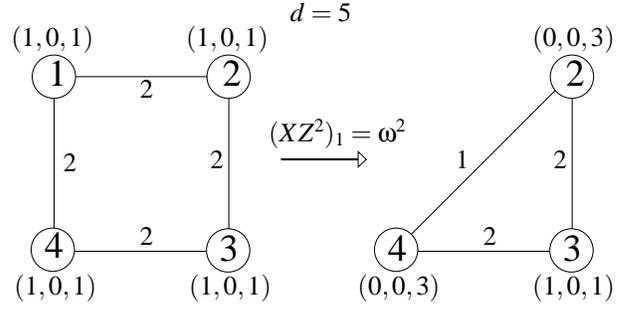}
\caption{Effect of a local $XZ^2$ measurement (with result $\omega^2$) on vertex 1 of an encoded-square qudit graph with weight-2 edges in $d=5$.
Labels are denoted $(z,x,m)$.  Following the measurement, vertex 1 and all its edges are removed.  Vertices 2 and 3 gain an edge of weight $N_{23}=0+6=1$ (modulo $d=5$)
and vertex labels $z=1+4+4+6=0$, $m=8=3$.}\label{fig-square}
\end{figure}

\subsection{Dependence and access}
We now consider conditions such that a subset of players can independently recover a secret encoded into the labels of a labelled qudit graph state,
using the terminology previously employed for the qubit case \cite{MS08}.
For a state $\projop{G_{\bm{\ell}}}$, consider a subset of players $\mathsf{V'} \subseteq \mathsf{V}$.  Without any assistance
from any other players the subset $\mathsf{V'}$ only has access to the reduced state $\rho_{\mathsf{V'}} = \Tr_{\mathsf{V}\setminus \mathsf{V'}} \projop{G_{\bm{\ell}}}$ and must recover the secret from this state.

Considering then some secret $s\in \F_d$, which is a function of the labels $\ell$, it follows that the set of players $\mathsf{V'}$ are clearly unable
to recover $s$ if the state $\rho_{\mathsf{V'}}$ is invariant under changes in $k$.  If this is not the case, we say the reduced state is {\it dependent}
on the value $s$.

However, dependency is generally insufficient for the players to be able to recover $s$.
We say that a set of players $\mathsf{V'}$  can {\it access} $s$ if there is a measurement
protocol the players in $\mathsf{V'}$ can perform on $\rho_{\mathsf{V'}}$ such that the value of $s$ is revealed
to them with certainty.

We find the following:

\begin{prop}
    \label{prop:access}
    Given an encoded graph state, if for a set of vertices $\mathsf{V'}\subseteq\mathsf{V}$,
    we have some corresponding set of values $\{w_i\}\in \F_d$ (with $w_i=0$ for all $\mathsf{v}_i\notin\mathsf{V'}$)
    then $\sum_i{w_iz_i}$ is
    accessible to the parties in $\mathsf{V'}$ via LOCC if $\sum_{\mathsf{v}_{j}\in\mathsf{V}}w_iA_{ij}=0$ for all
    $\mathsf{v}_i \in \mathsf{V} \setminus \mathsf{V'}$.
\end{prop}

\begin{proof}
    Each party $\mathsf{v}_i\in \mathsf{V'}$ measures  $K_i^{w_i}$ for an overall stabilizer
    $K = \prod_{\mathsf{v}_i \in \mathsf{V'}} K_i^{w_i}$. The graph state is thus
    an eigenstate of $K$ with eigenvalue $\omega^{-\sum_i w_iz_i}$; thus the parties
    in $\mathsf{V'}$  can jointly retrieve $\sum_i w_iz_i$.
\end{proof}

\begin{corr}
    \label{corr:access}
    In an encoded graph state, the label $z_k$ is accessible via LOCC by $\mathsf{V'}$
    if $\mathsf{V'}$ contains $\mathsf{v_k}$ and all of its neighbours.
\end{corr}

\begin{proof}
    Apply Proposition \ref{prop:access} with $w_k = 1$ and all $w_{i\ne k}=0$, so that $\sum_i w_iz_i = z_k$.
    For any $\mathsf{v}_j \in \mathsf{V} \setminus \mathsf{V'}$, we have that $\sum_{\mathsf{e}_{ij}\in\mathsf{e}}w_iA_{ij}=A_{kj}=0$
    since $\mathsf{v}_j$ is outside $\mathsf{V'}$ and thus not a
    neighbour of $\mathsf{v}_i$. The stabilizer $K$ in the proof of Proposition
    \ref{prop:access} is just $K_i$ in this case; hence measuring $K_i$ gives the value
    $z_i$.
\end{proof}

Note that, because labels are applied to a given vertex via {\it local} Pauli operations on that vertex, the reduced
graph state obtainable by tracing out a given vertex (i.e.\ a given player) must be independent
of that vertex's labels.  Hence those labels are inaccessible to the remaining players.  We exploit this
property to determine the minimum number of players required to recover a secret in threshold secret sharing protocols.

\subsubsection{Label ``shuffling''}
As for the qubit case, we demonstrate that a given subset of parties
is independent of a particular qudit label by ``shuffling'' the label onto
qudits not in the subset.  That is, we demonstrate that the transfer of the label
onto a different qudit (such that the original qudit no longer has any dependence on the label)
results in a physically equivalent state.  Thus the label shuffling procedure
does not represent any physical operation or change in the state: it merely generates
an equally valid relabelling of the same physical state.

\begin{prop}
    \label{prop:depend}
    Let $\ket{G_{\bm{\ell}}}$ be an encoded graph  state, and let players $i$ and $j$ be
    neighbours. Then $\ket{G_{\bm{\ell}}}$ is equivalent to the labelled graph
    state $\ket{G_{\bm{\ell}'}}$ where vertex $i$ is relabelled $z_i' = 0$,
    vertex $j$ is relabelled $(z_j', x_j') = (z_j, -A_{ij}^{-1} z_i)$, and each
    neighbour $k$ of $j$ is relabelled $z_k' = z_k - A_{ij}^{-1} A_{jk} z_i)$.
\end{prop}
\begin{proof}
Noting that $K_j$ stabilizes $\ket{G_{\bm{\ell}}}$, it follows that
    $K_j^{-A_{ij}^{-1} z_i} \ket{G_{\bm{\ell}}} \propto \ket{G_{\bm{\ell}'}}$, which, from (\ref{eq-kadef}), corresponds to the above relabelling.
\end{proof}

We demonstrate the above concept in Fig.\ \ref{fig-squareshuffle}
using a 4-qudit square graph with weight-2 edges in $d=5$, with a label of
$z=3$ on vertex 1 and all other labels equal to zero, and shuffling player 1's label to player 2.
\begin{figure}
\includegraphics[height=4cm]{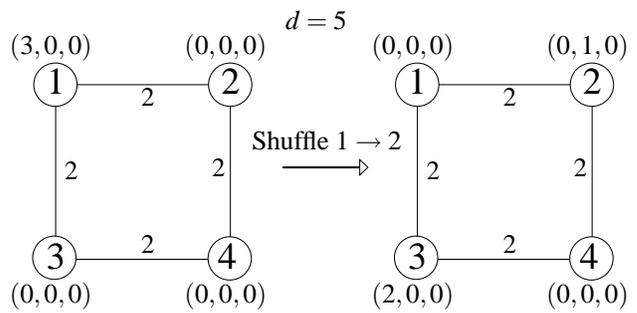}
\caption{Effect of label shuffling from vertex $1$ to $2$ on an encoded square qudit graph with weight-2 edges in $d=5$.  Labels are denoted $(z,x,m)$.
Vertex 1's $z$ label is set to 0.  Vertex 2's $x$ label is set to $-3/2=1$ (mod $d=5$), and vertex 3's $z$ label is set to $-6/2=2$.}\label{fig-squareshuffle}
\end{figure}

\section{Secret sharing protocols}\label{sec-protocols}
\subsection{CC scheme}
In a $(k,n)$ CC threshold scheme, the players' shared quantum state
is privately transferred to them by the dealer, after which
a classical secret can be
reconstructed by any $k$ of $n$ players (but no fewer) using classical
channels between the players.  We demonstrate the feasibility of such
schemes using labelled graph states by encoding a secret ``dit'' $s$ on to
the $z$ labels of certain vertices of the graph.  We then show that a subset
of $m$ parties may measure $s$ by jointly measuring one of the graph state's stabilizers $K_i$.
Finally, to demonstrate that fewer than $m$ parties cannot obtain $s$, it is sufficient to show
that we can shuffle the vertices' $z$ values such that any given set of fewer than $k$ parties has no dependence
on $s$ (for which it is clearly sufficient to demonstrate the above for $k-1$ parties only).

Note that secret sharing is known to be information-theoretically secure in the classical case; in our analysis
we are showing that such protocols may be unified within the graph state formalism.  CC protocols within our formalism
also provide a useful basis for the CQ and QQ cases, as shown later.

We find that such schemes may be
implemented for $d$-dimensional analogues of all the qubit states for which it was previously demonstrated \cite{MS08},
using essentially the same reasoning.

\subsubsection{(n,n) CC protocol in the tree state}\label{sec-CCnn}
We can implement an $(n,n)$ protocol using a tree state, consisting of a single vertex
connected by weight-1 edges to $n-1$ additional vertices, as illustrated in Fig.\ \ref{fig-nghzmcc}.
(In the 3-party qubit case with all labels set to zero this is the well-known Greenberger-Horne-Zeilinger-Mermin state).
We set $z_1=s$ and $z_i=0$ for all $i\ne 1$.
\begin{figure}
\includegraphics[height=4cm]{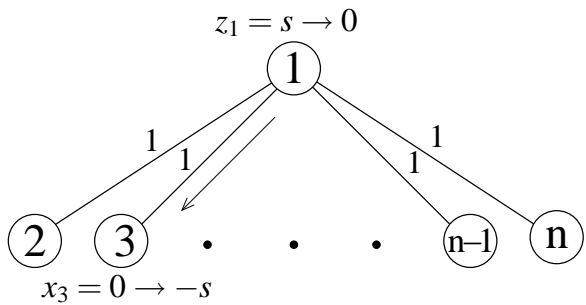}
\caption{Shuffling of a classical secret $s$ from player 1 to player 3 in a CC protocol for an encoded tree state.  As this label can be shuffled to any player
(or remain with player 1), any subset of the parties has no dependence on $s$.}\label{fig-nghzmcc}
\end{figure}
For this state, the stabilizer $K_1=(XZ^{m_1})_1\otimes Z_2\otimes\cdots \otimes Z_n$ has the eigenvalue
$\omega^{-s}$, and the value of $s$ may thus be accessed through local measurements by all $n$ parties.

We see that $s$ is inaccessible to any smaller number of parties as follows: The set of $n-1$ parties
excluding player 1 has no dependence on $s$.  However, we can also shuffle the $s$
dependence from player 1 to any other party $i$ as described in Proposition  \ref{prop:depend}.
This leaves us with $z_1=0$, $x_i=-z_1=-s$, and all other labels unchanged, in which case the labels
for all players other than $i$ have no dependence on $s$.  Hence all $n$ players are necessary to recover $s$.

\subsubsection{(2,3) CC protocol} \label{sec-CC32}

We can implement a $(2,3)$ protocol using the labelled graph state in Fig.\ \ref{fig-32cc}. We encode the classical secret $s\in \mathbb{F}_d$ by setting $z_1=0$,
 $z_2=2s$ and $z_3=s$ and set all other labels to zero as illustrated.
\begin{figure}
\includegraphics[height=3cm]{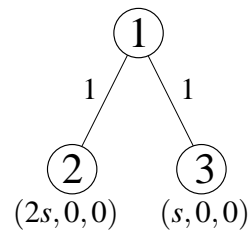}
\caption{ Encoding of a classical secret $s\in\mathbb{F}_d$ onto a three-party state to enable a $(2,3)$ CC protocol. It can be seen that the
secret can be shuffled away from any single qubit. At the same time any two players can access the secret as explained in the text.} \label{fig-32cc}
\end{figure}

To see that no single player can access the secret we apply the shuffling as described in Proposition 3. Player 1 cannot access the secret, as he has no label-dependence
on the secret. The label $z_2=2s$ can be removed by shuffling it to player 1 and player 3 as per Proposition 3 (explicitly by multiplication of
 stabilizer $K_1^{-2s}$); hence player 2 cannot access the secret. The label $z_3=s$ can be removed by shuffling it to player 1 and player 3 as per
 Proposition 3 (explicitly by multiplication of stabilizer $K_1^{-s}$), hence player 3 also cannot access the secret.

To see that any pair can access the secret, we take each pair separately.
Players 1 and 2 can measure $K_2 = {Z_1}X_2$, giving outcome $\omega^{-2s}$. Players 1 and 3 can measure $K_3 = {Z_1}X_3$, giving
outcome $\omega^{-s}$. Players 2 and 3 can measure $K_2^{d-1} K_3 = X_2^{d-1} X_3$, giving outcome $\omega^{s}$.

\subsubsection{(3,4) CC protocol in the ring state}
We can implement a $(3,4)$ protocol using a square encoded graph state with weight-1 edges and $m_i=0$ for all $i$,
by setting $z_i=s$ for all $i$.

Thus we have $K_i=\omega^{-s}$ for all $i$ and e.g. $K_2=Z_1X_2Z_3$, so players 1, 2 and 3 can collaborate
to measure $s$.  By the symmetry of the state, this is similarly the case for any set of three players.

As depicted in Fig.\ \ref{fig-43cc}, we can shuffle the $s$ dependence of player 1 to player 2.
This sets $z_1=0$, $x_2=-s$ and $z_3=s-s=0$.  This leaves the labels of 1 and 3
independent of $s$; hence neither they, nor (by symmetry) any two nonadjacent parties, can obtain $s$.
\begin{figure}
\includegraphics[height=4cm]{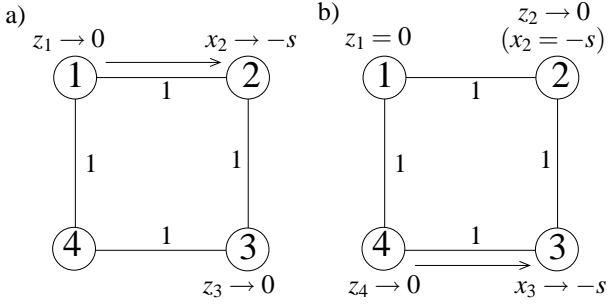}
\caption{CC protocol in an encoded $(3,4)$ ring state.  All players start with $z=s$. Dependence on the label $s$ can be shuffled such that any two a) nonadjacent
or b) adjacent parties are independent of the label.  Note that the the shuffling in b) occurs after first performing that in a)}\label{fig-43cc}
\end{figure}

If they subsequently perform a $4\to 3$ shuffle, this leaves player 1's labels unaltered and sets
player 4's $z$ label to 0; hence 1 and 4 (and thus any two adjacent players) cannot obtain $Z$, thereby completing
the proof.

\subsubsection{(3,5) CC protocol in the ring state}
We can implement a $(3,5)$ protocol using a 5-vertex ring-shaped encoded graph state with weight-1 edges, as shown in Figs.\ \ref{fig-53cca} and \ref{fig-53ccb},
setting $m_i=0$ for all $i$ and $z_i=s$ for all $i$, so again a $K_i$ measurement yields $\omega^{-s}$ for all $i$ and, for example, $K_2=Z_1X_2Z_3$ so players
1, 2 and 3 (hence any three adjacent parties) can measure $s$.

The other possibility is where only two of three players are adjacent, e.g. 1, 2 and 4.  For the ring state we
have
\begin{equation}
K_i=X_iZ_{i+1}Z_{i-1}
\end{equation}
so
\begin{equation}
K_1K_2K_4^{-1}= X_1 Z_2 X_2 Z_4 X^{-1}_4,
\end{equation}
which can be measured by players 1, 2 and 4, giving $\omega^{-s}$.  Hence any 3 parties can recover $s$.

\begin{figure}
\includegraphics[height=4cm]{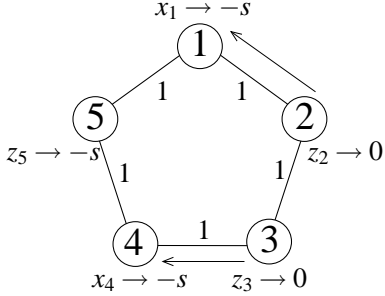}
\caption{CC protocol for an encoded $(3,5)$ ring state.  All players start with $z=s$. Dependence on the label $s$ can be shuffled such that any two adjacent
parties are independent of $s$ (see Fig.\ \ref{fig-53ccb} for nonadjacent parties.}\label{fig-53cca}
\end{figure}

\begin{figure}
\includegraphics[height=4cm]{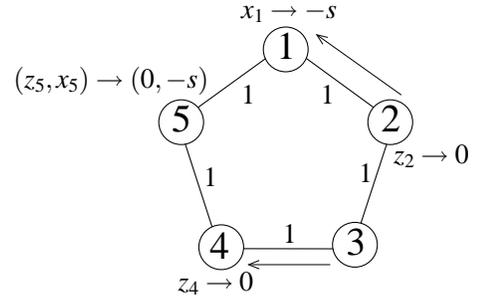}
\caption{Shuffling of the label $s$ such that any two nonadjacent parties
are independent of $s$, for the $(3,5)$ ring state.}\label{fig-53ccb}
\end{figure}
In the 5-qudit ring state we can always renumber the parties such that for
two parties $i$ and $j$, $i+1\le j \le i+2$.  In this case $i$ and $j$
can shuffle their $s$ dependence to the parties $i-1$ and $j+1$; hence two
players are not sufficient to recover $s$.

\subsection{The CQ scheme}
As in \cite{MS08}, we implement CQ schemes by performing a quantum key distribution (QKD) protocol,
wherein the dealer and a subset of the players end up securely sharing a {\it random} secret dit $s$.
This shared random classical information can then be used by these parties to securely share any {\it specified} secret of the same
length via standard classical cryptography protocols.

We implement such QKD protocols in graph states by having the dealer distribute the players' qudits to them
over public channels, after which the dealer and the players measure
their own qudits using bases chosen at random.  After measurement the dealer and players
publically announce their bases and only retain their measurement results for certain combinations
of chosen bases, such that $m$ of the players may collaborate to obtain a classical secret shared with the dealer.

The CQ protocols we have found follow a pattern whereby the result of the dealer's measurement
is a projection of the resultant reduced graph into a labelled state whose labels are correlated
with the dealer's measurement, and for which the players can perform a CC protocol.

As in \cite{MS08}, our protocol is based on that of Hillery, Bu{\v z}ek and Berthiaume \cite{HBB99} and
has the same limitation of being secure against certain classes of attack such as
intercept-resend attacks by an eavesdropper or certain classes of attack by dishonest participants \cite{HBB99},
but not against all attacks by dishonest participants \cite{KKI99}.  Additionally, we do not consider the case of
noisy channels, although existing protocols which do \cite{CL07} could potentially be adapted to our model.

\subsubsection{(n, n) CQ protocol in the extended tree state}
We can construct an $(n,n)$ CQ protocol using an ``extended'' tree state, consisting
of an $n$-party tree state (as in the CC protocol) with the dealer's qudit connected to player 1's qudit
by a weight-1 edge and all labels set to 0, as illustrated in Fig.\ \ref{fig-nghzmcq}.  This graph state therefore has stabilizers
\begin{align}
K_D&=X_D Z_1\\
K_1&=Z_D X_1\prod_{i\ne 1}Z_i\\
K_i&=X_iZ_1\qquad (i\ne 1).
\end{align}
\begin{figure}
\includegraphics[height=4cm]{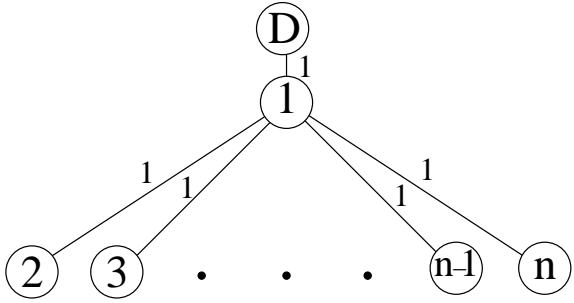}
\caption{The ``extended'' tree state used for the CQ protocol.  Local measurement on the dealer's qudit produces
a reduced state of the form of that in Fig.\ \ref{fig-nghzmcc} (although with an additional $m=t_D$ label on vertex 1.)}\label{fig-nghzmcq}
\end{figure}

In our protocol, the dealer and the players $1,\ldots ,n$ each independently randomly choose
a value $t_D, t_1,\ldots t_n \in \F_d$.  The dealer then measures her qudit in the basis
$X^{t_D}Z$, player 1 measures $XZ^{t_1}$ whereas the remaining players $i$ each measure $X^{t_i}Z$.

If the dealer's measurement result is $\omega^s$ then, from Proposition \ref{prop:measure2},
the reduced state of the remaining $n$ parties is the $n$GHZM state (without the additional dealer's vertex), still
with weight-1 edges.  The only relabelled vertex is that of player 1, with $z_1=s$, $m_1=t_D$.  All other labels remain zero.
The reduced graph therefore has stabilizers
\begin{align}
k_1&=(XZ^{t_D})_1 Z_2 \dotsb Z_n\\
k_i&= Z_1 X_i, \qquad i \ne 1
\end{align}
(here and henceforth we use $k_i$ to denote reduced-state stabilizers and $K_i$ for the full state including the dealer).
Furthermore, the labelling is such that the reduced state has eigenvalue 1 for all
the stabilizers with $i\ne 1$.  Thus any product of stabilizers $k_1k_2^{l_2}\dotsb k_n^{l_n}$
will have eigenvalue $\omega^{-s}$.  As
\begin{equation}
(XZ^t_1)_1\prod_{i\ne 1} (X^{t_i}Z)_i=k_1k_2^{t_2}\dotsb k_n^{t_n}Z^{\left(t_D-t_1+\sum_{i\ne 1}\right)}_1
\end{equation}
the players measure such a product (and hence are able to recover $s$) if $t_1 = t_D + t_2 + \dotsb + t_n$,
which occurs with probability $1/d$.

Security against an eavesdropper having tampered with the qudits (e.g. an intercept-resend attack)
is established by noting that the combination of the player and dealer's measurements satisfies
\begin{equation}
(X_{t_D}Z)_D(XZ^t_1)_1\prod_{i\ne 1} (X^{t_i}Z)_i=K_D^{t_D} K_1 K_2^{t_2}\dotsb K_n^{t_n};
\end{equation}
i.e.\ the players and dealer can also combine their measurement statistics for the various values of $t_i$
to measure the stabilizers of, and hence verify, the original state.

\subsubsection{(2, 3) CQ protocol} \label{sec-CQ32}

We can construct a $(2,3)$ CQ protocol through extending the CC case by attaching the dealer's qudit to those of players 2 and 3 with weights of two and one respectively and setting all labels to zero, as illustrated in Fig.~\ref{fig-23total}. The stabilizers for this state are
\begin{align}
K_D&=X_DZ^{2}_2Z_3,\\
K_1&=X_1Z_2Z_3,\\
K_2&=Z^{2}_DZ_1X_2,\\
K_3&=Z_DZ^{2}_1X_3.
\end{align}
\begin{figure}
\includegraphics[height=4cm]{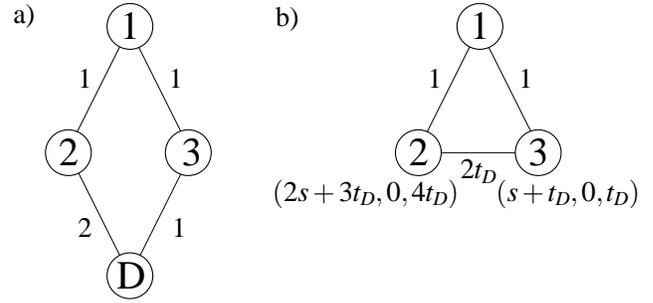}
\caption{a) The ``extended'' $(2,3)$ state used for the CQ protocol.  Local measurement on the dealer's qudit produces
a reduced state of the type in b)}\label{fig-23total}
\end{figure}
The dealer randomly selects a variable $t_D \in \F_d$ and measures $X^{t_D}Z$, and obtains result $\omega^s$. The resultant graph states are given by Proposition 5, as shown in Fig.~\ref{fig-23total}~b), with stabilizers
\begin{align} \label{eq-23stab}
k_{1,t_D}&=X_1Z_2Z_3,\\
k_{2,t_D}&=Z_1X_2Z_2^{4t_D}Z_3^{2t_D},\\
k_{3,t_D}&=Z_1Z_2^{2t_D}X_3Z_3^{t_D}.\label{eq-23stabend}
\end{align}

We now need to see how sets of players can or cannot access the random value $s$ of the dealer's measurement result. First, for any choice of $t_D$, no single player can access $s$ by application of Proposition 3 to Fig.~\ref{fig-23total} b).
To see how any pair can access the dealer's result, we must address each measurement parameter $t_D$ and resultant graph separately and see which measurement
 the players would ideally make. In the protocol, any given authorised set of players then choose randomly one of these possible measurements, which
 sometimes match that of the dealer and enable this set of players to access the secret and in addition provide security by effectively measuring a stabilizer of
 the total state Fig.~\ref{fig-23total} a).

We now go through how each pair can access $s$ if the dealer measures as above. Players $1$ and $2$ can access the dealer's result  by measuring
$$k_{1,t_D}^{-2t_D}k_{2,t_D},$$ allowing them to obtain $\omega^{-(2s+3t_D)}$. Players $1$ and $3$ would measure $$k_{1,t_D}^{-2t_D}k_{3,t_D},$$ allowing them to obtain $\omega^{-(s+t_D)}$.
Players $2$ and $3$ would measure $$k_{2,t_D}k_{3,t_D}^{-1},$$ allowing them to obtain $\omega^{-(s+2t_D)}$.

In this protocol, each authorised set will measure randomly from one of the three possible measurements corresponding to $t_D$.
For example, if players 1 and 2 decide to work together to obtain the secret, they will randomly measure $$k_{1,t_D'}^{-2t_D'}k_{2,t_D'}$$ by
 choosing their random variable $t_D'$ (agreed between players $1$ and $2$). By definition of the stabilizers these do not commute, hence cannot be
 measured simultaneously. Their measurements will coincide with the dealer's when $t_D'=t_D$, which happens with probability $1/d$. After announcing
 the bases they can then throw away all results that do not coincide.

Security can be proven by noting that including the dealer's measurement results allows the dealer plus an authorised set to simulate measurement
 of stabilizers of the original graph Fig.~\ref{fig-23total} a). An additional subtlety of simulating measurements locally in the qudit case,
 as apposed to the qubit case, is that there
 is more than one possible commuting Pauli.
Any local extended Pauli on a qudit $i$, $O_i$, commutes with all powers of this Pauli operator $O_i^m$, which means they share common eigenbases.
Further, as there are no degeneracies, the eigenvalues are just a re-ordering.
Thus, if measuring $O_i$ yields result $\omega^{s}$, we know that, even without doing any further measurement,
$O_i^m$ would give result $\omega^{ms}$. In this way, we can simulate any power of a Pauli $O^m_i$ just by measuring $O_i$.

In the protocol, some of the sifted key would be sacrificed now to check the security by announcing some randomly chosen results.
If we again look at players 1 and 2 for example, by combining their results with those of the dealer, they can simulate measuring stabilizers
$K_i$ of the original state (Fig.~\ref{fig-23total}a), equation (\ref{eq-23stab} - \ref{eq-23stabend}), since
\begin{align}
K_D^{2t_D}K_1^{-2t_D}K_2=\omega^{3t_D}(X_D^{t_D}Z_D)^2k_{1,t_D}^{-2t_D}k_{2,t_D}.
\end{align}
If there has been no interference from an eavesdropper, they will all measure outcome $1$ for, with high probability given a large number of announced results,
every value of $t_D$.
Although this outcome is not enough to guarantee that the state before
 measurement was the state of Fig.~\ref{fig-23total} a), it restricts the allowed space to that which guarantees security as follows: if $g$ is the subgraph
  got by removing player 3 and all its edges from Fig.~\ref{fig-23total} a), having eigenvalue $1$ for all these stabilizers implies that the state of the dealer
 plus players 1 and 2 is in the subspace spanned by $\{g_{\bm{z}=(i,i,0)}\}$ i.e.\ where players 1 and 2 have $z_1=z_2$ and players 3 has $z_3=0$.  Thus, the
most general state including an eavesdropper $E$ can be written
\begin{align}
\sum_i\alpha_i  \ket{i}_E\ket{g_{\bm{z}=(i,i,0)}}_{D,1,2},
\end{align}
It can then be checked that the reduced density matrix of the dealer and the eavesdropper is
\begin{align}
\rho_{ED}=\rho_E \otimes \one_D/d,
\end{align}
which implies that the eavesdropper cannot share any information about the dealer's results.
It can similarly be checked for pairs of players 2 and 3, and 1 and 3.

\subsubsection{(3,5) CQ protocol}
We can construct a $(3, 5)$ CQ protocol by extending the CC case by attaching the dealer's qudit to each player with weight one edges, and all labels set to zero, as illustrated in Fig.~\ref{fig-35a}. The stabilizers for this state are
\begin{align} \label{eq-35stab}
K_D&=X_DZ_1Z_2Z_3Z_4Z_5\\
K_1&=Z_DX_1Z_2Z_5\\
K_2&=Z_DZ_1X_2Z_3\\
K_3&=Z_DZ_2X_3Z_4\\
K_4&=Z_DZ_3X_4Z_5\\
K_5&=Z_DZ_1Z_4X_5\label{eq-35stabend}
\end{align}

\begin{figure}
\includegraphics[height=4cm]{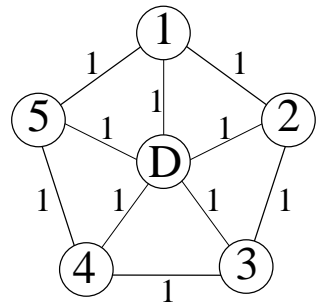}
\caption{The ``extended'' $(3,5)$ state used for the CQ protocol.  Local measurement on the dealer's qudit produces
a reduced state of the type in Fig.\ \ref{fig-35b}.}\label{fig-35a}
\end{figure}
\begin{figure}
\includegraphics[height=4cm]{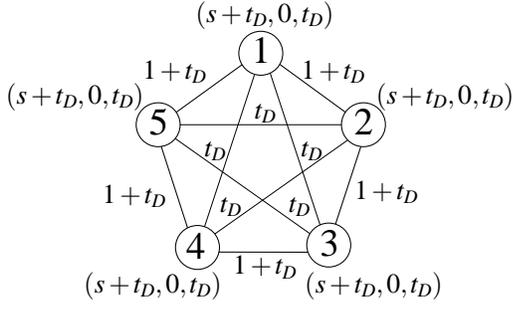}
\caption{Resultant state from the dealer's local measurement on Fig.\ \ref{fig-35a}.}\label{fig-35b}
\end{figure}

The dealer randomly selects a variable $t_D \in \F_d$ and measures $X^{t_D}Z$, and obtains result $\omega^s$. The resultant graph states are shown in Fig.~\ref{fig-35b}. The associated stabilizers of player $i$ are denoted $k_{i,t_D}$.

First we check that each pair of players cannot access $s$ without collaborating with the others. From Proposition 3, players 1 and 2 cannot access
the dealer's result (the secret) as explicitly
 by application of the stabilizers
\[
k_{4,t_D}^{-(1+st_D^{-1})}
\]
(shuffling from player 1 to 4) if $t_D\neq 0$ or
\[
k_{3,t_D}^{-s}k_{5,t_D}^{-s}
\]
(shuffling from 1 to 5 and 2 to 3) if $t_D=0$. We can
see that players 1 and 3 cannot access $s$ since the dependence can be removed by application of $$k_{2,t_D}^{-(s+t_D)(1+t_D)^{-1}}$$
 (shuffling from 1 to 2) if $t_D+1\neq 0$ or $$k_{4,t_D}^{-(1+st_D^{-1})}k_{5,t_D}^{-(1+st_D^{-1})}$$ (shuffling from 3 to 4 and 1 to 5)
 if $t_D+1=0$. By symmetry no pair alone can access $s$.

We now see how any three players can access the secret for a given $t_D$. Players 1, 2 and 3 can do so by measuring
\[
k_{1,t_D}^{-t_D}k_{2,t_D}^{1+2t_D}k_3^{-t_D},
\]
yielding result $\omega^{-(s+t_D)}$.
Players 1, 3 and 4 can do so by measuring
\[
k_{1,t_D}^{-(1+2t_D)} k_{3,t_D}^{1+t_D} k_{4,t_D}^{1+t_D},
\]
yielding result $\omega^{-(s+t_D)}$. By symmetry any three players can therefore find $s$.

In this protocol, any set of authorised players measures randomly in accordance with these procedures.
For example Players 1,2 and 3 will randomly measure stabilizers
\[
k_{1,t_D'}^{-t_D'}k_{2,t_D'}^{1+2t_D'}k_3^{-t_D'},
\]
by choosing a shared random variable $t_D'$. Security is then checked by verifying that the eigenequations
are satisfied by the dealer and the player, by communicating results to simulate measurement of a subset of stabilizers.
For example, the dealer plus players 1, 2 and 3 can, by comparing results, calculate what would have been the measured value for the stabilizers
(of the original graph Fig.~\ref{fig-35a}, equation (\ref{eq-35stab}-\ref{eq-35stabend})), because
\begin{align}
K_D^{t_D'}K_1^{-t_D'}&K_2^{1+2t_D'}K_3^{-t_D'}\\
=&\quad\omega^{t_D'}X_D^{t_D'}Z_Dk_{1,t_D'}^{-t_D'}k_{2,t_D'}^{1+2t_D'}k_3^{-t_D'}.
\end{align} 
If the eigenvalue is $1$ for all of these (as should be the case; if not the protocol restarts), this implies the joint state of the dealer
 plus players 1, 2 and 3 is in the subspace spanned by $\{ \ket{g_{\bm{z}=(i+j,i,0,j)}}_{D,1,2,3} \}$.
 As for the $(2,3)$ case, this implies that, for the most general state of an eavesdropper $E$ and dealer, their reduced density matrix is of the form
\begin{align}
\rho_{E,D}= \rho_E \otimes \one_D / d;
\end{align}
hence the eavesdropper's state is completely uncorrelated with the dealer's measurement result and therefore uncorrelated with the secret.
Security can similarly be checked for all sets of three or more players.
\subsection{The QQ scheme}
The QQ scheme proposed in \cite{MS08} is readily generalisable to qudits.
This protocol is similar to the CQ scheme, but the secret to be shared is now a
quantum state $\ket{s}$ in a $d$-dimensional Hilbert space, initially possessed by the dealer,
who distributes it to the other parties via a joint operation
on the secret state and the parties' shared graph state, in a manner analogous
to quantum teleportation.
We describe the general protocol explicitly below.

Denoting the dealer's secret qudit as 
\begin{equation}
\ket{s}_D = \sum_{i=0}^{d-1} \alpha_i \ket{i}_D,\label{eq-qqdq}
\end{equation}
the dealer prepares the state
\begin{equation}
\ket{s}_D \ket{G}_{D,V}
\end{equation}
corresponding to some graph state $G$ for the dealer's qudit $D$ and all the players'
qudits $V$.  The dealer distributes the players' qudits to them (over perfect quantum channels,
which we can assume to be public).  The dealer then measures her two qudits in the generalized
Bell basis \cite{PhysRevLett.70.1895} $\{ \ket{\psi_{mn}}\}$, where
\begin{equation}
\ket{\psi_{mn}} \defeq \frac{1}{\sqrt{d}} \sum_j \omega^{jn} \ket{j} \ket{j + m}.
\end{equation}

If the dealer's measurement result is $(m,n)$, corresponding to the state $\ket{\psi_{mn}}$,
then it follows from our rules for projective measurement that
the resultant state for all parties is
\begin{align}
&\ket{\psi_{mn}}_D\bra{\psi_{mn}} \ket{s}_D \ket{G}_{D,V}\nonumber\\
 \propto\quad&\ket{\psi_{mn}}_{D} \sum_j \alpha_j \omega^{-jn} \ket{g_{\bm{z}=(j+m)(A_{D1},A_{D2}\ldots,A_{DN})}}_V
\end{align}
where $\ket{g_{\bm{ z}}}$ is the encoded reduced graph state on the players $1,
\dotsc, n$ with labels $\bm{z}$.

If the dealer informs the players of their measurement result $(m,n)$, then
a set of players $\in V$ (including some player $a$ with $N_{Da} \ne 0$) can apply
a correction operator
\begin{equation}
 U_{mn} \defeq K_a^{-nN_{Da}^{-1}} Z^{-m \bm{A_D}}
\end{equation}
to obtain the state
\begin{equation}
\ket{s_g}^V = \sum_j \alpha_j \ket{g_{\bm z = j(A_{D1},\ldots A_{DN})}}_V.
\end{equation}
The access properties of this final state depend on the graph state used.  Qualitatively,
for certain initial graph states, the state $\ket{s_g}^V$ can be regarded as a superposition (with coefficients corresponding
to those of the encoded secret) of orthogonal labelled graph states whose labels have the same access structure as CC protocols.
Thus, the ability to recover the quantum secret corresponds to the ability to recover these classical labels, providing
a natural extension of the classical protocols to the quantum case.

However, the above reasoning does not guarantee that such a QQ protocol will have the same adversary structure
as the corresponding CC protocol, as seen below, where we achieve ``perfect'' threshold schemes in the $(2,3)$
and $(5,3)$ cases, but fewer than $n$ players can still obtain some secret information in the $(n,n)$ case.

\subsubsection{$(n,n)$ QQ protocol using the $n$GHZM state}
If the joint graph state of the dealer and players $\ket{G}$
is the same as for the CQ scheme (i.e.\ an $n$GHZM state with an additional dealer's vertex,
connected to player 1, and all labels set to zero), then the final players' state is
\begin{equation}
\ket{s_g}_V = \sum_j \alpha_j \ket{g_{\bm{z} = (j,0,\ldots,0)}}_V
\end{equation}
where $V$ is the set of all players.
The $\ket{g}$ states, in the form of $n$GHZM states with a label on player 1's vertex,
are the same as those used in the $(n,n)$ CC protocol, to which the access properties
described in Sec.\ \ref{sec-CCnn} apply, and fewer than $n$ players cannot perfectly
reconstruct the secret.

However, we note that the reduced states of fewer than $n$ players are not
independent of the $\alpha_j$.  For example, in the analogous qubit protocol \cite{MS08}
we can, in some basis (depending on the player), express the reduced state of a single player as
$\rm{diag}(|\alpha_0+\alpha_1|^2,|\alpha_0-\alpha_1|^2)$.
Thus this $(n,n)$ protocol is not a ``perfect'' threshold scheme; it allows some
incomplete information to be accessed by fewer than $n$ players.

However the secret qudit is clearly encoded within the state of the $n$ players and so may
be accessed by all players by performing joint quantum operations.  Furthermore,
the secret can be isolated to any individual player through LOCC,as follows.

For any player $i$ other than player 1, the secret can be isolated by
player 1 measuring $X$ and the rest of the players other than $i$ measuring $Z$,
which is evident by rewriting the state as
\begin{equation}
\ket{s_g}_V = \sum_j \ket{g_{\bm{z} = (j,0,\ldots,0)}}_{\mathsf{V} \setminus \mathsf{v}_i} \Bigl(
\sum_k\alpha_k \ket{j-k}_i \Bigr).
\end{equation}

The secret can be isolated to player 1 by every player other than 1 measuring
$Z$, thereby producing the state
\begin{align}
 \ket{s_g'}_V =\quad&\sum_{k_2} \dotsb \sum_{k_n} \sum_j \alpha_j\nonumber\\
&\times\ket{\overline{-(j+k_2+\dotsb+k_n)}}_1 \ket{k_2}_2 \dotsb \ket{k_n}_n.
\end{align}

Thus the quantum secret requires collaboration by all $n$ players (and no fewer) to recover but
the players may collaboratively make the secret locally accessible to a single player
if they so wish.

\subsubsection{(2,3) QQ protocol} \label{sec-QQ32}
We can construct a $(2,3)$ QQ protocol by teleporting the secret qudit (\ref{eq-qqdq}) into the extended graph state of Fig.~\ref{fig-23total}. The encoded state is
\begin{align}
\ket{s_g}_V = \sum_{j=0}^{d-1} \alpha_j  \ket{g_{\bm{z}=j(0,2,1)}}\label{eq-qq32}
\end{align}
for $g$ the (unlabelled) graph in Fig.~\ref{fig-32cc}.  The state (\ref{eq-qq32}) can be rewritten as
\begin{align}
\ket{s_g}_V &=  \frac{1}{\sqrt{d}}\sum_{j,k=0}^{d-1} \alpha_j  \ket{k}_1\ket{\overline{2j+k}}_2\ket{\overline{j+k}}_3\nonumber\\
&=\frac{1}{\sqrt{d}}\sum_{j,M=0}^{d-1} \alpha_j  \ket{M-j}_1\ket{\overline{M+j}}_2\ket{\overline{M}}_3
\end{align}
where $M=j+k$.

For $d=3$ this state corresponds exactly to the $(2,3)$ qutrit code presented in \cite{CGL99}. Any pair of players can thus access the secret,
and any single player cannot. For example, players 1 and 2 perform control operations that add the value of the first register (without an overbar)
to the second register (with an overbar), and then
$-\frac{1}{2}$ times the value of the second to the first (again, keeping the first register non-barred and the second always barred). This leaves the resultant state as
\begin{align}
\frac{1}{\sqrt{d}}\sum_{j=0}^{d-1}\alpha_j\ket{-j}_1\sum_{M=0}^{d-1}\ket{\overline{2M}}_2\ket{\overline{M}}_3.
\end{align}
Any pair can similarly obtain the secret.  Since single players form complementary sets to pairs of players, it follows from information gain implying disturbance \cite{CGL99} (and can be shown explicitly) that no single player can obtain any information about the secret, so this is a perfect
threshold scheme.

\subsubsection{(3, 5) QQ protocol} \label{sec-QQ53}

We can construct a $(3,5)$ QQ protocol by teleporting the secret qudit into the extended graph state of Fig.~\ref{fig-35a}. The encoded state is then
\begin{align}
\ket{s_g}_V = \sum_{j=0}^{d-1} \alpha_j\ket{g_{\bm{z}=(j,j,\ldots j)}}
\end{align}
where $g$ is the (unlabelled) graph of Fig.~\ref{fig-53cca}.
The state can be rewritten as
\begin{align}
\ket{s_g}_V =\quad&\frac{1}{d} \sum_{x,y = 0}^{d-1}\ket{x}_1\ket{y}_2\nonumber\\
&\times\left(\sum_{j=0}^{d-1}\alpha_j\omega^{j(x+y)+xy}\ket{h_{\bm{z}=(j+x,j,j+y)}}_{3,4,5}\right)
\end{align}
where $h$ is the subgraph corresponding to players 3, 4 and 5 taken by removing all other vertices and their edges. The fact that each vector $\bm{z}$
occurs only once means that players 3, 4 and 5 can always measure in such a way to project them only onto the a subspace corresponding to some $|xy\rangle_{1,2}$.
Each would occur with equal probability of $1/d^2$. 

If for example their measurements put them onto the $\ket{h_{\bm{z}=(j,j,j)}}_{3,4,5}$ subspace (corresponding to $x=y=0$),
they know from the above that they would obtain the state
\begin{align}
\sum_{j=0}^{d-1} \alpha_j\ket{h_{\bm{z}=(j,j,j)}}_{3,4,5},
\end{align}
hence would have the secret encoded into the graph state $|h\rangle$. Suitable global operations could then be used to map the state onto any subsystem they liked.
For example performing $(C_{ab})^{-1}$ between players $3$ and $4$, as well as between players $4$ and $5$, gives the state
\begin{align}
\sum_{j=0}^{d-1} \alpha_j |\bar{j}\rangle_3|\bar{j}\rangle_4|\bar{j}\rangle_{5}.
\end{align}
Players $3$ and $4$ can measure in their respective $Z$ bases, to shuffle the information to player 5. Similarly it can be shuffled to any player.

To see how players 2, 4 and 5 can access the secret qutrit, we rewrite the state as
\begin{align}
\ket{s_g}_V =\quad&\frac{1}{d}\sum_{x,y = 0}^{d-1}\ket{x}_1\ket{y}_2\nonumber\\
&\times\left(\sum_{j=0}^{d-1}\alpha_j\omega^{j(x+y)}\ket{h_{\bm{z}=(j+x+y,j+x,j+y)}}_{2,4,5}\right)
\end{align}
where $h$ is the subgraph of players 2, 4 and 5 found by removing all other vertices and their edges. Each possible vector $\bm{z}$ occurs only once in this expansion.
By the same tactics players 2,4 and 5 can access the secret. By symmetry, any three can access the secret.  As with the $(2,3)$ case, perfect
secret recovery by 3 players requires that the complementary sets of 2 players cannot obtain any information about the secret (as can also be shown explicitly), so this is a perfect
threshold scheme.

We note that the existence of a $(3,5)$ QQ protocol in graph states of arbitrary dimension has been implied by the work of Gheorghiu, Looi and Griffiths \cite{GLG09},
who showed that encoded information can be isolated to 3 of the 5 qudits in a graph state.

\section{Conclusion}
We have developed a unified framework for the three types of secret sharing in prime-dimensional Hilbert space.  Our formalism uses prime-dimensional
graph states, analogous to those in the qubit case, and our protocols allow sharing of classical ``dits'' and prime-dimensional quantum states as secrets.
Our protocols include the case of $(2,3)$ sharing, which was not achievable in the qubit formalism in addition to higher-dimensional analogues of qubit
protocols.  Our work provides a useful step towards describing schemes sharing classical and quantum secrets of any size within a graph-state formalism for general physical
systems.

Our work suggests several areas for further research.  Still open are the questions of whether any achievable threshold secret sharing scheme
can be constructed within our formalism, and, if so, what the associated graph states and protocols would be.  If prime-dimensional systems
are not sufficient to construct any scheme, it is possible that an analogous formalism involving states of any dimension might allow for a
richer array of protocols, as we have shown occurs when moving beyond the qubit case.  Generalising to the case of arbitrary $d$ is thus
another promising area for further work, and some of the methods described in \cite{GLG09} (which considers such general graph states) may be useful here.

With regard to achieving arbitrary threshold schemes, we additionally note that the graph state formalism can in principle be used for sharing of quantum secrets
for arbitrary access structures. This property of the formalism stems from the fact that schemes exist for arbitrary access structures using high-dimensional stabilizer states
 \cite{CGL99,Gott00}, and that it can be shown that all stabilizer states are locally equivalent to qudit graph states \cite{Schlingemann01,BB06}.
 However there is no explicit procedure for converting these schemes to graphs, and further, it is not immediately clear that they can also be used
 for sharing classical secrets (though it seems very likely).  Developing such a procedure would provide a graph state formalism for arbitrary
access structures in quantum secret sharing, although it may not be equivalent to that in our current work.

There are also several possible extensions to the analysis of CQ protocols within our formalism: considering more general attacks involving dishonest participants,
the presence of channel noise, and the overall efficiency of the protocol; i.e. the need to discard states when the participants' bases do not match.  It is an
open question which graph states would produce the highest key rates (and hence secret-sharing rates) per state within our formalism.

Further generalisation may allow a wide range of secret-sharing schemes to be described in this intuitively-appealing way.\\

We thank Gilad Gour for valuable discussions on the limitations of the $(n,n)$ QQ protocol.
This work was supported by NSERC, QuantumWorks, MITACS, and USARO.  DM acknowledges financial support from ANR
Projet FREQUENCY (ANR-09-BLAN-0410-03).  BF is partially supported by a 
Pacific
Institute for Mathematical Sciences Postdoctoral Fellowship, and BCS
is partially supported by a Canadian Institute for Advanced Research
Fellowship.

\bibliographystyle{apsrev4-1}
\bibliography{gqudits}

\end{document}